\documentclass{llncs}
\usepackage{mathpartir}
\usepackage{amsmath,amssymb,mathtools,xspace,upgreek}
\usepackage{float}

\floatstyle{boxed} 
\restylefloat{figure}

\usepackage{etapsBisim}

\title{Applicative Bisimulations\\ for Delimited-Control Operators}

\author{Dariusz Biernacki \and Sergue\"i Lenglet \thanks{The author is supported by the
    Alain Bensoussan Fellowship Programme.}}

\institute{University of Wroc\l{}aw}

\begin{document}

\pagestyle{plain}

\maketitle

\thispagestyle{plain}

\begin{abstract}
  We develop a behavioral theory for the untyped call-by-value
  $\lambda$-calculus extended with the delimited-control operators
  shift and reset. For this calculus, we discuss the possible
  observable behaviors and we define an applicative bisimilarity that
  characterizes contextual equivalence. We then compare the
  applicative bisimilarity and the CPS equivalence, a relation on
  terms often used in studies of control operators. In the process, we
  illustrate how bisimilarity can be used to prove equivalence of
  terms with delimited-control effects.
\end{abstract}

\section{Introduction}

Morris-style contextual equivalence \cite{JHMorris:PhD} is usually
regarded as the most natural behavioral equivalence for functional
languages based on $\lambda$-calculi. Roughly, two terms are
equivalent if we can exchange one for the other in a bigger program
without affecting its behavior (\ie whether it terminates or not). The
quantification over program contexts makes contextual equivalence hard
to use in practice and, therefore, it is common to look for more
effective characterizations of this relation. One approach is to rely
on coinduction, by searching for an appropriate notion of
\emph{bisimulation}. The bisimulation has to be defined in such a way
that its resulting behavioral equivalence, called \emph{bisimilarity},
is \emph{sound} and \emph{complete} with respect to contextual
equivalence (\ie it is included and contains contextual equivalence,
respectively).

The problem of finding a sound and complete bisimilarity in the
$\lambda$-calculus has been well studied and usually leads to the
definition of an \emph{applicative} bisimilarity
\cite{Abramsky-Ong:IaC93,Howe:IaC96,Gordon:TCS99} (or, more recently,
\emph{environmental} bisimilarity \cite{Sangiorgi-al:LICS07}). The
situation is more complex in $\lambda$-calculi extended with
\emph{control operators} for first-class continuations---so far, only
a few works have been conducted on the behavioral theory of such
calculi. A first step can be found for the $\lambda\mu$-calculus (a
calculus that mimics abortive control operators such as {\it
  call/cc}~\cite{Parigot:LPAR92}) in \cite{Bierman:MFCS98} and
\cite{David-Py:JSL01}, where it is proved that the definition of
contextual equivalence can be slightly simplified by quantifying over
evaluation contexts only; such a result is usually called a
\emph{context lemma}. In \cite{Stoevring-Lassen:POPL07}, St{\o}vring
and Lassen define an \emph{eager normal form bisimilarity} (based on
the notion of L{\'e}vy-Longo tree
equivalence)~\cite{Lassen:LICS05,Lassen:MFPS05,Lassen-Levy:LICS08}
which is sound for the $\lambda\mu$-calculus, and which becomes sound
and complete when a notion of state is added to the
$\lambda\mu$-calculus. In~\cite{Merro-Biasi:06}, Merro and Biasi
define an applicative bisimilarity which characterizes contextual
equivalence in the \emph{CPS calculus}~\cite{Thielecke:PHD}, a minimal
calculus which models the control features of functional languages
with imperative jumps. As for the $\lambda$-calculus extended with
control only, however, no sound and complete bisimilarities have been
defined.

In this article, we present a sound and complete applicative
bisimilarity for a $\lambda$-calculus extended with Danvy and
Filinski's static de\-li\-mi\-ted-control operators {\it shift} and
{\it reset} \cite{Danvy-Filinski:LFP90}. In contrast to abortive
control operators, delimited-control operators allow to delimit access
to the current continuation and to compose continuations. The
operators shift and reset were introduced as a direct-style
realization of the traditional success/failure continuation model of
backtracking otherwise expressible only in continuation-passing
style. The numerous theoretical and practical applications of shift
and reset (see, e.g.,~\cite{Biernacka-al:LMCS05} for an extensive
list) include the seminal result by Filinski showing that a
programming language endowed with shift and reset is monadically
complete~\cite{Filinski:POPL94}.

The $\lambda$-calculi with static delimited-control operators have
been an active research topic from the semantics as well as type- and
proof-theoretic point of view~(see, \eg
\cite{Biernacka-al:LMCS05,Biernacka-Biernacki:PPDP09,Ariola-al:HOSC07a}).
However, to our knowledge, no work has been carried out on the
behavioral theory of such $\lambda$-calculi. In order to fill this
void, we present a study of the behavioral theory of an untyped,
call-by-value $\lambda$-calculus extended with shift and reset
\cite{Danvy-Filinski:LFP90}, called $\lamshift$. In Section
\ref{s:lamshift}, we give the syntax and reduction semantics of
$\lamshift$, and discuss the possible observable behaviors for the
calculus. In Section \ref{s:bisimilarity}, we define an applicative
bisimilarity, based on a labelled transition semantics, and prove it
characterizes contextual equivalence, using an adaptation of Howe's
congruence proof method \cite{Howe:IaC96}. As a byproduct, we also
prove a context lemma for $\lamshift$.  In Section
\ref{s:cps-equivalence}, we study the relationship between applicative
bisimilarity and an equivalence based on translation into
continuation-passing style (CPS), a relation often used in works on
control operators and CPS. In the process, we show how applicative
bisimilarity can be used to prove equivalence of terms. Section
\ref{s:conclusion} concludes the article and gives ideas for future
work. The appendices contain the proofs missing from the body of the
article.

\section{The Language $\lamshift$}
\label{s:lamshift}

In this section, we present the syntax, reduction semantics, and
contextual equivalence of the language $\lamshift$ used throughout
this article.

\subsection{Syntax}

The language $\lamshift$ extends the call-by-value $\lambda$-calculus
with the delimited-control operators \emph{shift} and
\emph{reset}~\cite{Danvy-Filinski:LFP90}. We assume we have a set of
term variables, ranged over by $\varx$ and $\vark$. We use two
metavariables to distinguish term variables bound with a
$\lambda$-abstraction from variables bound with a shift; we believe
such distinction helps to understand examples and reduction rules. The
syntax of terms and values is given by the following grammars:
$$
\begin{array}{llll}
  \textrm{Terms:} &\quad \tm & ::= & \varx \Mid \lam \varx \tm \Mid \app \tm \tm \Mid
  \shift \vark \tm \Mid \reset \tm \\
  \textrm{Values:}&\quad \val & ::= & \lam \varx \tm
\end{array}
$$
The operator \emph{shift} ($\shift \vark \tm$) is a capture operator,
the extent of which is determined by the delimiter \emph{reset}
($\rawreset$). A $\lambda$-abstraction $\lam \varx \tm$ binds $\varx$ in
$\tm$ and a shift construct $\shift \vark \tm$ binds $\vark$ in $\tm$;
terms are equated up to $\alpha$-conversion of their bound
variables. The set of free variables of $\tm$ is written $\fv \tm$; a
term is \emph{closed} if it does not contain any free
variable. Because we work mostly with closed terms, we consider only
$\lambda$-abstractions as values.

We distinguish several kinds of contexts, defined below, which all can
be seen as terms with a hole.
$$
\begin{array}{llll}
  \textrm{Pure evaluation contexts:}&\quad \ctx & ::= & \mtctx \Mid \vctx \val \ctx \Mid \apctx \ctx
  \tm \\
  \textrm{Evaluation contexts:}&\quad \rctx & ::= & \mtctx \Mid \vctx \val \rctx \Mid \apctx \rctx
  \tm \Mid \resetctx \rctx \\
  \textrm{Contexts:}&\quad \cctx & ::= & \mtctx \Mid \lam \varx \cctx
  \Mid \vctx \tm \cctx \Mid \apctx \cctx \tm \Mid \shift \vark \cctx \Mid \resetctx
  \cctx 
\end{array}
$$
Regular contexts are ranged over by $\cctx$. The pure evaluation
contexts\footnote{This terminology comes from Kameyama (\eg in
  \cite{Kameyama-Hasegawa:ICFP03}).} (abbreviated as pure contexts),
ranged over by $\ctx$, represent delimited continuations and can be
captured by the shift operator. The call-by-value evaluation contexts,
ranged over by $\rctx$, represent arbitrary continuations and 
encode the chosen reduction strategy. Following the correspondence
between evaluation contexts of the reduction semantics and control
stacks of the abstract machine for shift and reset, established by
Biernacka et al.~\cite{Biernacka-al:LMCS05},
we interpret contexts inside-out, \ie $\mtctx$ stands for the empty
context, $\vctx \val \ctx$ represents the ``term with a hole'' $\inctx
\ctx {\app \val \hole}$, $\apctx \ctx \tm$ represents $\inctx \ctx
     {\app \hole \tm}$, $\resetctx \rctx$ represents $\inctx \rctx
     {\reset \hole}$, etc. (This choice does not affect the results
     presented in this article in any way.) Filling a context $\cctx$
     (respectively $\ctx$, $\rctx$) with a term $\tm$ produces a term,
     written $\inctx \cctx \tm$ (respectively $\inctx \ctx \tm$,
     $\inctx \rctx \tm$); the free variables of $\tm$ can be captured
     in the process. A context is \emph{closed} if it contains only
     closed terms.

\subsection{Reduction Semantics}
\label{ss:reduction}

Let us first briefly describe the intuitive semantics of shift and
reset by means of an example written in SML using Filinski's
implementation of shift and reset~\cite{Filinski:POPL94}.
\begin{example}
  \label{e:copy}
The following function copies a list~\cite{Biernacki-al:RS-05-16} (the
SML expression {\tt shift (fn k => t)} corresponds to $\shift k t$
  and {\tt reset (fn () => t)} corresponds to $\reset t$):
\begin{quote}
\begin{verbatim}
fun copy xs = 
  let fun visit nil = nil
        | visit (x::xs) = visit (shift (fn k => x :: (k xs)))
  in reset (fn () => visit xs) end
\end{verbatim}
\end{quote}
This program illustrates the main ideas of programming with shift and
reset:
\begin{itemize}
\item Reset delimits continuations. Control effects are local to {\tt
  copy}.
\item Shift captures delimited continuations. Each, but last,
  recursive call to {\tt visit} abstracts the continuation {\tt
    fn v => reset (fn () => visit v)} and binds it to {\tt k}.
\item Captured continuations are statically composed. When applied in
  the expression {\tt k xs}, the captured continuation becomes the
  current delimited continuation that is isolated from the rest of the
  program by a control delimiter---witness the reset expression in the
  captured continuation.
\end{itemize}
\end{example}

Formally, the call-by-value reduction semantics of $\lamshift$ is
defined by the following rules, where $\subst \tm \varx \val$ is the
usual capture-avoiding substitution of $\val$ for $\varx$ in $\tm$:
$$
\begin{array}{lrll}
  \RRbeta & \quad \inctx \rctx {\app {\lamp \varx \tm} \val} & \redcbv & \inctx
  \rctx {\subst \tm \varx \val} \\
  \RRshift & \quad \inctx \rctx {\reset{\inctx \ctx {\shift \vark \tm}}} &
  \redcbv & \inctx \rctx{\reset{\subst \tm \vark
    {\lam \varx {\reset {\inctx \ctx \varx}}}}} \mbox{ with } \varx \notin \fv
  \ctx\\
  \RRreset & \quad \inctx \rctx {\reset \val} & \redcbv & \inctx \rctx \val
\end{array}
$$ 
The term $\app {\lamp \varx \tm} \val$ is the usual call-by-value
redex for $\beta$-reduction (rule $\RRbeta$). The operator $\shift
\vark \tm$ captures its surrounding context $\ctx$ up to the
dynamically nearest enclosing reset, and substitutes $\lam \varx
{\reset {\inctx \ctx \varx}}$ for $\vark$ in $\tm$ (rule
$\RRshift$). If a reset is enclosing a value, then it has no purpose
as a delimiter for a potential capture, and it can be safely removed
(rule $\RRreset$). All these reductions may occur within a metalevel
context $\rctx$. The chosen call-by-value evaluation strategy is
encoded in the grammar of the evaluation contexts.

\begin{example}
  \label{e:reduction}
  Let $i = \lam \varx \varx$ and $\omega = \lam \varx {\app \varx \varx}$. We
  present the sequence of reductions initiated by $\reset {\app{\appp {(\shift
        {\vark_1}{\app i {\appp {\vark_1} i}})}{\shift {\vark_2} \omega}}{\appp
      \omega \omega}}$. The term $\shift {\vark_1}{\app i {\appp {\vark_1} i}}$
  is within the pure context $\ctx = \apctx {(\apctx \mtctx {\appp \omega
      \omega})}{\shift {\vark_2} \omega}$ (remember that we represent contexts
  inside-out), enclosed in a delimiter $\rawreset$, so $\ctx$ is captured
  according to rule $\RRshift$.
  $$\reset {\app{\appp {(\shift{\vark_1}{\app i {\appp {\vark_1} i}})}{\shift
        {\vark_2} \omega}}{\appp \omega \omega}} \redcbv \reset{\app i {\appp
      {\lamp \varx {\reset {\app{\appp \varx {\shift {\vark_2} \omega}}{\appp
              \omega \omega}}}} i}}$$
  The role of reset in $\lam \varx {\reset{\inctx \ctx \varx}}$ becomes
  clearer after reduction of the $\RRbeta$-redex $\app {\lamp \varx {\reset{\inctx \ctx
      \varx}}} i$.
  $$\reset{\app i {\appp {\lamp \varx {\reset {\app{\appp \varx {\shift {\vark_2}
                \omega}}{\appp \omega \omega}}}} i}} \redcbv \reset {\app i
    {\reset{\app {\appp i {\shift {\vark_2} \omega}}{\appp \omega \omega}}}}$$
  When the captured context $\ctx$ is reactivated, it is not \emph{merged} with
  the context $\vctx i \mtctx$, but \emph{composed} thanks to the
  reset enclosing $\ctx$. As a result, the capture triggered by $\shift
  {\vark_2} \omega$ leaves the term $i$ outside the first enclosing reset
  untouched. 
  $$\reset {\app i {\reset{\app {\appp i {\shift {\vark_2}
            \omega}}{\appp \omega \omega}}}} \redcbv \reset {\app i
    {\reset \omega}}$$ Because $\vark_2$ does not occur in $\omega$,
  the context $\vctx i {(\apctx \mtctx {\appp \omega \omega})}$ is
  discarded when captured by $\shift {\vark_2} \omega$. Finally, we
  remove the useless delimiter $\reset {\app i {\reset \omega}}
  \redcbv \reset{\app i \omega}$ with rule $\RRreset$, and we then
  $\RRbeta$-reduce and remove the last delimiter $\reset{\app i
    \omega} \redcbv \reset{\omega} \redcbv \omega$. Note that, while
  the reduction strategy is call-by-value, some function arguments are
  not evaluated, like the non-terminating term $\app \omega \omega$
  in this example.  
\end{example}

There exist terms which are not values and which cannot be reduced any
further; these are called \emph{stuck terms}.
\begin{definition}
  A closed term $\tm$ is stuck if $\tm$ is not a value and $\tm \not \redcbv$.
\end{definition}
For example, the term $\inctx \ctx {\shift \vark \tm}$ is stuck
because there is no enclosing reset; the capture of $\ctx$ by the
shift operator cannot be triggered. In fact, closed stuck terms are
easy to characterize.
\begin{lemma}
  \label{l:stuck}
  A closed term $\tm$ is stuck iff $\tm = \inctx \ctx {\shift \vark
    {\tm'}}$ for some $\ctx$, $k$, and $\tm'$.
\end{lemma}
We call \emph{redexes} (ranged over by $\redex$) the terms of the form
$\app{\lamp \varx \tm} \val$, $\reset {\inctx \ctx {\shift \vark
    \tm}}$, and $\reset \val$. Thanks to the following
unique-decomposition property, the reduction $\redcbv$ is
deterministic.
\begin{lemma}
  For all closed terms $\tm$, either $\tm$ is a value, or it is a
  stuck term, or there exist a unique redex $\redex$ and a unique
  context $\rctx$ such that $\tm = \inctx \rctx \redex$.
\end{lemma}
 
Given a relation $\rel$ on terms, we write $\rtclo\rel$ for the
transitive and reflexive closure of $\rel$. We define the evaluation
relation of $\lamshift$ as follows.
\begin{definition}
  We write $\tm \evalcbv \tm'$ if $\tm \clocbv \tm'$ and $\tm' \not\redcbv$. 
\end{definition}
The result of the evaluation of a closed term, if it exists, is either
a value or a stuck term. If a term $\tm$ admits an infinite reduction
sequence, we say it \emph{diverges}, written $\tm \divcbv$. In the
rest of the article, we use extensively $\Omega = \app {\lamp \varx
  {\app \varx \varx}}{\lamp \varx {\app \varx \varx}}$ as an example
of such a term.

\subsection{Contextual Equivalence}
\label{ss:context-equiv}

In this section, we discuss the possible definitions of a Morris-style
contextual equivalence for the calculus $\lamshift$. As usual, the
idea is to express that two terms are equivalent iff they cannot be
distinguished when put in an arbitrary context. The question is then
what kind of behavior we want to observe. As in the regular
$\lambda$-calculus we could observe only termination (\ie does a term
reduce to a value or not), leading to the following relation.

\begin{definition}
  Let $\tmzero$, $\tmone$ be closed terms. We write $\tmzero \ctxequivone
  \tmone$ if for all closed $\cctx$, $\inctx \cctx \tmzero \evalcbv \valzero$
  implies $\inctx \cctx \tmone \evalcbv \valone$, and conversely for $\inctx
  \cctx \tmone$.
\end{definition}
This definition does not mention stuck terms; as a result, they can be
equated with diverging terms. For example, let $\tmzero = \app
{(\shift \vark {\app \vark {\lam \varx \varx})}} \Omega$, $\tmone =
\Omega$, and $\cctx$ be a closed context. If $\inctx \cctx \tmzero
\evalcbv \valzero$, then we can prove that for all closed $\tm$, there
exists $\val$ such that $\inctx \cctx \tm \evalcbv \val$ (roughly,
because $\tm$ is never evaluated; see Lemma~\ref{l:aone} in
Appendix~\ref{a:context-equiv}). In particular, we have $\inctx \cctx
\tmone \evalcbv \valone$. Hence, we have $\tmzero \ctxequivone
\tmone$.

A more fine-grained analysis is possible, by observing stuck terms. 
\begin{definition}
  Let $\tmzero$, $\tmone$ be closed terms. We write $\tmzero \ctxequivtwo
  \tmone$ if for all closed $\cctx$,
  \begin{itemize}
  \item $\inctx \cctx \tmzero \evalcbv \valzero$ implies $\inctx
    \cctx \tmone \evalcbv \valone$;
  \item $\inctx \cctx \tmzero \evalcbv \tmzero'$, where $\tmzero'$
    is stuck, implies $\inctx \cctx \tmone \evalcbv \tmone'$, with
    $\tmone'$ stuck as well;
  \end{itemize}
  and conversely for $\inctx \cctx \tmone$.
\end{definition}
The relation $\ctxequivtwo$ distinguishes the terms $\tmzero$ and $\tmone$
defined above. We believe $\ctxequivtwo$ is more interesting because it gives
more information on the behavior of terms; consequently, we use it as the
contextual equivalence for $\lamshift$. Henceforth, we simply write $\ctxequiv$
for $\ctxequivtwo$.

The relation $\ctxequiv$, like the other equivalences on terms defined in this
article, can be extended to open terms in the following way.
\begin{definition}
  Let $\rel$ be a relation on closed terms. The \emph{open extension}
  of $\rel$, written $\open\rel$, is defined on open terms as: we write
  $\tmzero \open\rel \tmone$ if for every substitution $\subs$ which
  closes $\tmzero$ and $\tmone$, $\tmzero \subs \rel \tmone \subs$
  holds.
\end{definition}

\begin{remark}
  Contextual equivalence can be defined directly on open terms by requiring that
  the context $\cctx$ binds the free variables of the related terms. The
  resulting relation would be equal to $\open\ctxequiv$ \cite{Gordon:TCS99}.
\end{remark}

\section{Bisimilarity for $\lamshift$}
\label{s:bisimilarity}

In this section, we define an applicative bisimilarity and prove it
equal to contextual equivalence.

\subsection{Labelled Transition System}

To define the bisimilarity for $\lamshift$, we propose a labelled
transition system (LTS), where the possible interactions of a term
with its environment are encoded in the labels. Figure \ref{f:lts}
defines a LTS $\tmzero \lts\act \tmone$ with three kinds of
transitions. An \emph{internal action} $\tm \lts\tau \tm'$ is an
evolution from $\tm$ to $\tm'$ without any help from the surrounding
context; it corresponds to a reduction step from $\tm$ to $\tm'$. The
transition $\valzero \lts \valone \tm$ expresses the fact that
$\valzero$ needs to be applied to another value $\valone$ to evolve,
reducing to $\tm$. Finally, the transition $\tm \lts \ctx \tm'$ means
that $\tm$ is stuck, and when $\tm$ is put in a context $\ctx$
enclosed in a reset, the capture can be triggered, the result of which
being $\tm'$.

\begin{figure}[t]
\begin{mathpar}
  \inferrule{ }
  {\app {\lamp \varx \tm} \val \ltstau \subst \tm \varx \val}~\LTSbeta
  \and
  \inferrule{ }
  {\reset \val \ltstau \val}~\LTSreset
  \and
  \inferrule{\tmzero \ltstau \tmzero'}
  {\app \tmzero \tmone \ltstau \app {\tmzero'} \tmone}~\LTScompl
  \and
  \inferrule{\tm \ltstau \tm'}
  {\app \val \tm \ltstau \app \val {\tm'}}~\LTScompr
  \and
  \inferrule{\tm \ltstau \tm'}
  {\reset \tm \ltstau \reset {\tm'}}~\LTScompreset
  \and
  \inferrule{\tm \lts\mtctx \tm'}
  {\reset \tm \ltstau \tm'}~~\LTScaptreset
  \and
  \inferrule{ }
  {\lam \varx \tm \lts\val \subst \tm \varx \val}~\LTSval
  \and
  \inferrule{\varx \notin \fv \ctx}
  {\shift \vark \tm \lts\ctx \reset {\subst \tm \vark {\lam \varx {\reset
          {\inctx \ctx \varx}}}}}~\LTSshift
  \and
  \inferrule{\tmzero \lts{\apctx \ctx \tmone} \tmzero'}
  {\app \tmzero \tmone \lts\ctx \tmzero'}~\LTScaptl
  \and
  \inferrule{\tm \lts{\vctx \val \ctx} \tm'}
  {\app \val \tm \lts\ctx \tm'}~\LTScaptr
\end{mathpar}
\caption{Labelled Transition System}
\label{f:lts}
\end{figure}

Most rules for internal actions (Fig. \ref{f:lts}) are
straightforward; the rules $\LTSbeta$ and $\LTSreset$ mimic the
corresponding reduction rules, and the compositional rules
$\LTScompr$, $\LTScompl$, and $\LTScompreset$ allow internal actions to happen
within any evaluation context. The rule $\LTScaptreset$ for context capture
is explained later. Rule $\LTSval$ defines the only possible
transition for values. Note that while both rules $\LTSbeta$ and
$\LTSval$ encode $\beta$-reduction, they are quite different in
nature; in the former, the term $\app{\lamp \varx \tm} \val$ can
evolve by itself, without any help from the surrounding context, while
the latter expresses the possibility for $\lam \varx \tm$ to evolve
only if a value $\val$ is provided by the environment.

The rules for context capture are built following the principles of
complementary semantics developed in \cite{Lenglet-al:CONCUR09}. The
label of the transition $\tm \lts\ctx \tm'$ contains what the
environment needs to provide (a context $\ctx$, but also an enclosing
reset, left implicit) for the stuck term $\tm$ to reduce to
$\tm'$. Hence, the transition $\tm \lts\ctx \tm'$ means that we have
$\reset {\inctx \ctx \tm} \lts\tau \tm'$ by context capture. For
example, in the rule $\LTSshift$, the result of the capture of $\ctx$
by $\shift \vark \tm$ is $\reset{\subst \tm \vark {\lam \varx {\reset
      {\inctx \ctx \varx}}}}$.

In rule $\LTScaptl$, we want to know the result of the capture of
$\ctx$ by the term $\app \tmzero \tmone$, assuming $\tmzero$ contains
an operator shift. Under this hypothesis, the capture of $\ctx$ by
$\app \tmzero \tmone$ comes from the capture of $\apctx \ctx \tmone$
by $\tmzero$. Therefore, as premise of the rule $\LTScaptl$, we check
that $\tmzero$ is able to capture $\apctx \ctx \tmone$, and the result
$\tmzero'$ of this transition is exactly the result we want for the
capture of $\ctx$ by $\app \tmzero \tmone$. The rule $\LTScaptr$
follows the same pattern. Finally, a stuck term $\tm$ enclosed in a
reset is able to perform an internal action (rule $\LTScaptreset$); we
obtain the result $\tm'$ of the transition $\reset \tm \lts\tau \tm'$
by letting $\tm$ capture the empty context, \ie by considering the
transition $\tm \lts\mtctx \tm'$.

\begin{example}
  With the same notations as in Example \ref{e:reduction}, we illustrate how
  the LTS handles capture by considering the transition from $\reset{\app
    {\appp i {\shift \vark \omega}}{\appp \omega \omega}}$. 
  \begin{mathpar}
    \inferrule*[Right=\LTScaptreset]{ 
      \inferrule*[Right=\LTScaptl]{
        \inferrule*[Right=\LTScaptr]{
          \inferrule*[Right=\LTSshift]{ }
          {\shift \vark \omega \lts{\vctx i {(\apctx \mtctx {\appp \omega
                  \omega})}} \reset \omega}}
        {
          \app i {\shift \vark \omega} \lts{\apctx \mtctx {\appp \omega
              \omega}} \reset \omega}}
      {
        \app{\appp i {\shift \vark \omega}}{\appp \omega \omega} \lts \mtctx 
        \reset \omega}}
    {\reset{\app {\appp i {\shift \vark \omega}}{\appp \omega \omega}} \lts\tau
      \reset \omega}
  \end{mathpar}
  Reading the tree from bottom to top, we see that the rules
  $\LTScaptreset$, $\LTScaptl$, and $\LTScaptr$ build the captured
  context in the label by deconstructing the initial term. Indeed, the
  rule $\LTScaptreset$ removes the outermost reset, and initiates the
  context in the label with $\mtctx$. The rules $\LTScaptl$ and
  $\LTScaptr$ then successively remove the outermost application and
  store it in the context. The process continues until a shift
  operator is found; then we know the captured context is completed,
  and the rule $\LTSshift$ computes the result of the capture. This
  result is then simply propagated from top to bottom by the other
  rules.  
\end{example}

The LTS corresponds to the reduction semantics and exhibits the observable
terms (values and stuck terms) of the language in the following way. 

\begin{lemma}
  \label{l:lts-redcbv-main}
  The following hold:
  \begin{itemize}
  \item We have $\ltstau \,=\, \redcbv$.
  \item If $\tm \lts\ctx \tm'$, then $\tm$ is a stuck term, and $\reset{\inctx
      \ctx \tm} \ltstau \tm'$.
  \item If $\tm \lts\val \tm'$, then $\tm$ is a value, and $\app \tm \val
    \ltstau \tm'$.
  \end{itemize}
\end{lemma}

\subsection{Applicative Bisimilarity}

We now define the notion of applicative bisimilarity for $\lamshift$. We write
$\cloltstau$ for the reflexive and transitive closure of $\ltstau$. We define
the weak delay\footnote{where internal steps are allowed before, but not after a
  visible action} transition $\ltswk\act$ as $\cloltstau$ if $\act=\tau$ and as
$\cloltstau \lts\act$ otherwise. The definition of the (weak delay) bisimilarity
is then straightforward.
\begin{definition}
  A relation $\rel$ on closed terms is an applicative simulation if $\tmzero
  \rel \tmone$ implies that for all $\tmzero \lts\act \tmzero'$, there exists
  $\tmone'$ such that $\tmone \ltswk\act \tmone'$ and $\tmzero' \rel \tmone'$.

  A relation $\rel$ on closed terms is an applicative bisimulation if $\rel$
  and $\inv\rel$ are simulations. Applicative bisimilarity $\bisim$ is the
  largest applicative bisimulation.
\end{definition}
In words, two terms are equivalent if any transition from one is matched by a
weak transition with the same label from the other. As in the $\lambda$-calculus
\cite{Abramsky-Ong:IaC93,Gordon:TCS99}, it is not mandatory to test the
internal steps when proving that two terms are bisimilar, because of the
following result.
\begin{lemma}
  \label{l:bisim-eval}
  If $\tm \ltstau \tm'$ (respectively $\tm \evalcbv \tm'$) then $\tm \bisim
  \tm'$.
\end{lemma}
Lemma \ref{l:bisim-eval} holds because $\{ (\tm, \tm')\,,\,\tm \ltstau \tm' \}$
is an applicative bisimulation. Consequently, applicative bisimulation can be
defined in terms of big-step transitions as follows.

\begin{definition}
  A relation $\rel$ on closed terms is a big-step applicative simulation if
  $\tmzero \rel \tmone$ implies that for all $\tmzero \ltswk\act \tmzero'$ with
  $\act \neq \tau$, there exists $\tmone'$ such that $\tmone \ltswk\act \tmone'$
  and $\tmzero' \rel \tmone'$. 

  A relation $\rel$ on closed terms is a big-step applicative bisimulation if
  $\rel$ and $\inv\rel$ are big-step applicative simulations. Big-step
  applicative bisimilarity $\apbisim$ is the largest big-step applicative
  bisimulation.
\end{definition}
Henceforth, we drop the adjective ``applicative'' and refer to the two kinds of
relations simply as ``bisimulation'' and ``big-step bisimulation''.
\begin{lemma}
  \label{l:apbisim-bisim}
  We have $\bisim\, =\, \apbisim$.
\end{lemma}
The proof is by showing that $\bisim$ is a big step bisimulation, and
that $\apbisim$ is a bisimulation (using a variant of Lemma
\ref{l:bisim-eval} involving $\apbisim$). As a result, if $\rel$ is a
big-step bisimulation, then $\rel \, \subseteq \, \apbisim \,
\subseteq \, \bisim$. We work with both styles (small-step and
big-step), depending on which one is easier to use in a given proof.

\begin{example}
  Assuming we add lists and recursion to the calculus, we informally
  prove that the function {\tt copy} defined in Example \ref{e:copy}
  is bisimilar to its effect-free variant, defined below.
  \begin{quote}
\begin{verbatim}
fun copy2 nil = nil
  | copy2 (x::xs) = x::(copy2 xs)
\end{verbatim}
  \end{quote}
  To this end, we define the relations (where we let $l$ range over lists, and
  $e$ over their elements)
  \begin{eqnarray*}
    \rel_1 & = & \{ 
    (\reset {e_1:: \reset {e_2 :: \ldots \reset {e_n :: \reset {{\tt visit}
            ~l}}}},\,
    e_1::(e_2:: \ldots e_n::({\tt copy2}~l))) \}\\
    \rel_2 & = & \{(\reset {e_1:: \reset {e_2 :: \ldots \reset {e_n :: \reset l}}},\,           
    e_1::(e_2:: \ldots e_n::l)) \}
  \end{eqnarray*}
  and we prove that $\rel_1 \cup \rel_2 \cup \{ (l,\, l)\}$ is a
  bisimulation. First, let $\tmzero \rel_1 \tmone$. If $l$ is
  empty, then both {\tt visit}~$l$ and {\tt copy2}~$l$ reduce to
  the empty list, and we obtain two terms related by
  $\rel_2$. Otherwise, we have $l=e_{n+1}::l'$, $\reset {{\tt
      visit}~l}$ reduces to $\reset {e_{n+1}::\reset{{\tt
        visit}~l'}}$, ${\tt copy2}~l$ reduces to $e_{n+1}::({\tt
    copy2}~l')$, and therefore $\tmzero$ and $\tmone$ reduce to terms
  that are still in $\rel_1$. Now, consider $\tmzero \rel_2
  \tmone$; the transition from $\tmzero$ removes the delimiter
  surrounding $l$, giving a term related by $\rel_2$ to $\tmone$ if
  there are still some delimiters left, or equal to $\tmone$ if all
  the delimiters are removed. Finally, two identical lists are clearly
  bisimilar.
\end{example}

\subsection{Soundness}

To prove soundness of $\bisim$ w.r.t. contextual equivalence, we show
that $\bisim$ is a congruence using \emph{Howe's method}, a well-known
congruence proof method initially developed for the $\lambda$-calculus
\cite{Howe:IaC96,Gordon:TCS99}. We briefly sketch the method and explain
how we apply it to $\bisim$; the complete proof can be found in
Appendix \ref{a:howe}.

The idea of the method is as follows: first, prove some basic properties of
\emph{Howe's closure} $\clohbisim$, a relation which contains $\bisim$ and is a
congruence by construction. Then, prove a simulation-like property for
$\clohbisim$. From this result, prove that $\clohbisim$ and $\bisim$ coincide on
closed terms. Because $\clohbisim$ is a congruence, it shows that $\bisim$ is a
congruence as well. The definition of $\clohbisim$ relies on the notion of
\emph{compatible refinement}; given a relation $\rel$ on open terms, the
compatible refinement $\comp\rel$ relates two terms iff they have the same
outermost operator and their immediate subterms are related by $\rel$. Formally,
it is inductively defined by the following rules.
\begin{mathpar}
  \inferrule{ }
  {\varx \comp\rel \varx}
  \and
  \hspace{-0.1em}\inferrule{\tmzero \rel \tmone}
  {\lam \varx \tmzero \comp\rel \lam \varx \tmone}
  \and
  \inferrule{\tmzero \rel \tmone \\ \tmzero' \rel \tmone'}
  {\app \tmzero {\tmzero'} \comp\rel \app \tmone {\tmone'}}
  \and
  \inferrule{\tmzero \rel \tmone}
  {\shift \vark \tmzero \comp\rel \shift \vark \tmone}
  \and
  \inferrule{\tmzero \rel \tmone}
  {\reset \tmzero \comp\rel \reset\tmone}
\end{mathpar}
Howe's closure $\clohbisim$ is inductively defined as the smallest
congruence containing $\open\bisim$ and closed under right composition
with $\open\bisim$.

\begin{definition}
  Howe's closure $\clohbisim$ is the smallest relation satisfying: 
  \begin{mathpar}
    \inferrule{\tmzero \open\bisim \tmone}
    {\tmzero \clohbisim \tmone}
    \and
    \inferrule{\tmzero \clohbisim\open\bisim \tmone}
    {\tmzero \clohbisim \tmone}
    \and
    \inferrule{\tmzero \comp\clohbisim \tmone}
    {\tmzero \clohbisim \tmone}
  \end{mathpar}
\end{definition}
By construction, $\clohbisim$ is a congruence (by the third rule of the definition), and composing on the
right with $\open\bisim$ gives some transitivity properties to
$\clohbisim$. In particular, it helps in proving the following
classical results (see \cite{Gordon:TCS99} for the proofs).
\begin{lemma}[Basic properties of $\clohbisim$]
  \label{l:properties}
  The following hold:
  \begin{itemize}
  \item For all $\tmzero$, $\tmone$, $\valzero$, and $\valone$, $\tmzero
    \clohbisim \tmone$ and $\valzero \clohbisim \valone$ implies $\subst \tmzero
    \varx {\valzero} \clohbisim \subst \tmone \varx {\valone}$.
  \item The relation $\rtclo{(\clohbisim)}$ is symmetric.
  \end{itemize}
\end{lemma}
The first item states that $\clohbisim$ is substitutive. This property
helps in establishing the simulation-like property of $\clohbisim$
(second step of the method). Let $\clohbisimc$ be the restriction of
$\clohbisim$ to closed terms. We cannot prove directly that
$\clohbisimc$ is a bisimulation, so we prove a stronger result
instead. We extend $\clohbisim$ to labels, by defining $\ctx
\clohbisim \ctx'$ as the smallest congruence extending $\clohbisim$
with the relation $\mtctx \clohbisim \mtctx$, and by adding the
relation $\tau \clohbisim \tau$.
\begin{lemma}[Simulation-like property]
  \label{l:sim-property-main}
  If $\tmzero \clohbisimc \tmone$ and $\tmzero \lts\act \tmzero'$, then for all
  $\act \clohbisimc \act'$, there exists $\tmone'$ such that $\tmone
  \ltswk{\act'} \tmone'$ and $\tmzero' \clohbisimc \tmone'$.
\end{lemma}
Using Lemma \ref{l:sim-property-main} and the fact that $\rtclo{(\clohbisimc)}$
is symmetric (by the second item of Lemma~\ref{l:properties}), we can prove that
$\rtclo{(\clohbisimc)}$ is a bisimulation. Therefore, we have
$\rtclo{(\clohbisimc)} \,\subseteq\, \bisim$, and because $\bisim \,\subseteq\,
\clohbisimc \,\subseteq\, \rtclo{(\clohbisimc)}$ holds by construction, we can
deduce $\bisim \,=\, \clohbisimc$. Because $\clohbisimc$ is a congruence, we have
the following result.
\begin{theorem}
  The relation $\bisim$ is a congruence.
\end{theorem}
As a corollary, $\bisim$ is sound w.r.t. contextual equivalence.
\begin{theorem}
  \label{t:soundness-main}
  We have $\bisim \,\subseteq\, \ctxequiv$.
\end{theorem}

\subsection{Completeness and Context Lemma}

In this section, we prove that $\bisim$ is complete w.r.t. $\ctxequiv$. To this
end, we use an auxiliary relation $\rctxequiv$, defined below, which refines
contextual equivalence by testing terms with evaluation contexts only. While
proving completeness, we also prove $\rctxequiv \,=\, \ctxequiv$, which means that
testing with evaluation contexts is as discriminative as testing with any
contexts. Such a simplification result is similar to Milner's context lemma
\cite{Milner:TCS77}.
\begin{definition}
  Let $\tmzero$, $\tmone$ be closed terms. We write $\tmzero \rctxequiv
  \tmone$ if for all closed $\rctx$,
  \begin{itemize}
  \item $\inctx \rctx \tmzero \evalcbv \valzero$ implies $\inctx \rctx \tmone
    \evalcbv \valone$; 
  \item $\inctx \rctx \tmzero \evalcbv \tmzero'$, where $\tmzero'$ is stuck,
    implies $\inctx \rctx \tmone \evalcbv \tmone'$, with $\tmone'$ stuck as
    well;
  \end{itemize}
  and conversely for $\inctx \rctx \tmone$.
\end{definition}
Clearly we have $\ctxequiv \,\subseteq\, \rctxequiv$ by definition. The relation
$\bisim$ is complete w.r.t. $\rctxequiv$.
\begin{theorem}
  \label{t:completeness-main}
  We have $\rctxequiv \,\subseteq\, \bisim$.
\end{theorem}
The proof of Theorem \ref{t:completeness-main} is the same as in
$\lambda$-calculus \cite{Gordon:TCS99}; we prove that $\rctxequiv$ is a
big-step bisimulation, using Lemmas \ref{l:lts-redcbv-main}, \ref{l:bisim-eval},
and Theorem \ref{t:soundness-main}. The complete proof can be found in Appendix
\ref{a:completeness}. We can now prove that all the relations defined so far coincide.

\begin{theorem}
  We have $\ctxequiv \,=\, \rctxequiv \,=\, \bisim$.
\end{theorem}
Indeed, we have $\rctxequiv \,\subseteq\, \bisim$ (Theorem
\ref{t:completeness-main}), $\bisim \,\subseteq\, \ctxequiv$ (Theorem
\ref{t:soundness-main}), and $\ctxequiv \,\subseteq\, \rctxequiv$ (by
definition). 

\section{Relation to CPS Equivalence}
\label{s:cps-equivalence}
In this section we study the relationship between our bisimilarity
(and thus contextual equivalence) and an equivalence relation based on
translating terms with shift and reset into continuation-passing style
(CPS). Such an equivalence has been characterized in terms of
direct-style equations by Kameyama and Hasegawa who developed an
axiomatization of shift and reset~\cite{Kameyama-Hasegawa:ICFP03}. We
show that all but one of their axioms are validated by the
bisimilarity of this article, which also provides several examples of
use of the bisimilarity. We also pinpoint where the two relations
differ.

\subsection{Axiomatization of Delimited Continuations}

The operators shift and reset have been originally defined by a
translation into continu\-ation-passing
style~\cite{Danvy-Filinski:LFP90} that we present in
Fig.~\ref{f:cps-translation}. Translated terms expect two
continuations: the delimited continuation representing the rest of the
computation up to the dynamically nearest enclosing delimiter and the
metacontinuation representing the rest of the computation beyond this
delimiter.

It is natural to relate any other theory of shift and reset to their
definitional CPS translation. For example, the reduction rules $\tm
\redcbv \tm'$ given in Section \ref{ss:reduction} are sound w.r.t. the
CPS because CPS translating $\tm$ and $\tm'$ yields $\beta
\eta$-convertible terms in the $\lambda$-calculus. More generally, the
CPS translation for shift and reset induces the following notion of
equivalence on terms:

\begin{figure}[t]
  $$
  \begin{array}{rcl}
    \cps{\varx} &=& \lam {\vark_1 \vark_2}{\app {\app {\vark_1} \varx} \vark_2} \\
    \cps{\lam \varx \tm} &=& \lam {\vark_1 \vark_2}{\app{\app {\vark_1}{\lamp
          \varx {\cps \tm}}}{\vark_2}} \\
    \cps{\app \tmzero \tmone} &=& \lam {\vark_1 \vark_2}{\app{\app {\cps \tmzero}{\lamp
          {\varx_0 \vark'_2}{\app{\app {\cps \tmone}{\lamp{\varx_1
                  \vark''_2}{\app{\app{\app
                    {\varx_0}{\varx_1}}{\vark_1}}{\vark''_2}}}}{\vark'_2}}}}{\vark_2}} \\
    \cps{\reset \tm} &=& \lam {\vark_1 \vark_2}{\app{\app {\cps \tm}
        \kinit}{\lamp \varx {\app {\app {\vark_1} \varx}{\vark_2}}}} \\
    \cps{\shift \vark \tm} &=& \lam {\vark_1 \vark_2}{\app{\app {\subst {\cps
            \tm} \vark {\lamp{\varx_1 \vark'_1 \vark'_2}{\app{\app
                {\vark_1}{\varx_1}}{\lamp {\varx_2}{\app{\app
                    {\vark'_1}{\varx_2}}{\vark'_2}}}}}} \kinit}{\vark_2}}\\
    \mbox{with } \kinit &=& \lam {\varx \vark_2}{\app {\vark_2} \varx}
  \end{array}
  $$
\caption{CPS translation}
\label{f:cps-translation}
\end{figure}

\begin{definition}
  Terms $\tm$ and $\tm'$ are CPS equivalent if their CPS
  translations are $\beta\eta$-convertible.
\end{definition}

In order to relate the bisimilarity of this article and the CPS
equivalence, we use Kameyama and Hasegawa's axioms
\cite{Kameyama-Hasegawa:ICFP03}, which characterize the CPS
equivalence in a sound and complete way: two terms are CPS equivalent
iff one can derive their equality using the equations of
Fig. \ref{f:axioms}. Kameyama and Hasegawa's axioms relate not only
closed, but arbitrary terms and they assume variables as values.

\begin{figure}[t]
$$
\begin{array}{rcllrcll}
  \app{\lamp \varx \tm} \val &=& \subst \tm \varx \val & \quad \AXbeta 
  & \app {\lamp \varx {\inctx \ctx \varx}} \tm &=& \inctx \ctx \tm \mbox{ if } x
  \notin \fv \ctx & \quad \AXbetaomega \\
  \reset {\inctx \ctx {\shift \vark \tm}} &=& \reset {\subst \tm \vark {\lam
      \varx {\reset {\inctx \ctx \varx}}}} & \quad \AXresetshift
  & \quad  \reset {\app {\lamp \varx \tmzero}{\reset \tmone}} &=& \app {\lamp \varx
    {\reset \tmzero}}{\reset \tmone} & \quad \AXresetlift  \\
  \reset \val &=& \val & \quad \AXresetval 
  & \shift \vark {\reset \tm} &=& \shift \vark \tm & \quad \AXshiftreset \\
  \lam \varx {\app \val \varx} &=& \val \mbox{ if } \varx \notin \fv \val &
  \quad \AXetav 
  & \shift \vark {\app \vark \tm} &=& \tm \mbox{ if } \vark \notin \fv
  \tm & \quad \AXshiftelim \\
\end{array}
$$
\caption{Axiomatization of $\lamshift$}
\label{f:axioms}
\end{figure}

\subsection{Kameyama and Hasegawa's Axioms through Bisimilarity}
\label{ss:example}

We show that closed terms related by all the axioms except for
$\AXshiftelim$ are bisimilar. In the following, we write $\Id$ for the
bisimulation $\{ (\tm, \tm) \}$.

\begin{proposition}
  We have $\app{\lamp \varx \tm} \val \bisim \subst \tm \varx \val$, $\reset
  {\inctx \ctx {\shift \vark \tm}} \bisim \reset {\subst \tm \vark {\lam \varx
      {\reset {\inctx \ctx \varx}}}}$, and $\reset \val \bisim \val$.
\end{proposition}

\begin{proof}
  These are direct consequences of the fact that $\lts\tau \,\subseteq\, \bisim$
  (Lemma \ref{l:bisim-eval}).
  \qed
\end{proof}

\begin{proposition}
  If $\varx \notin \fv \val$, then $\lam \varx {\app \val \varx} \bisim \val$.
\end{proposition}

\begin{proof}
  We prove that $\rel\; = \{(\lam \varx {\app {\lamp y \tm} \varx}, \lam y \tm),
  \varx \notin \fv \tm \} \cup \bisim$ is a bisimulation. To this end, we have
  to check that $\lam \varx {\app {\lamp y \tm} \varx} \lts\valzero \app {\lamp
    y \tm} \valzero$ is matched by $\lam y \tm \lts\valzero \subst \tm y
  \valzero$, \ie that $\app {\lamp y \tm} \valzero \rel \subst \tm y \valzero$
  holds for all $\valzero$. We have $\app {\lamp y \tm} \valzero \lts\tau \subst
  \tm y \valzero$, and because $\lts\tau \,\subseteq\, \bisim \,\subseteq\, \rel$, we
  have the required result.  \qed
\end{proof}

\begin{proposition}
  \label{p:reset-reset}
  We have $\shift \vark {\reset \tm} \bisim \shift \vark \tm$. 
\end{proposition}

\begin{proof}
  Let $\rel\; = \{ (\reset {\reset \tm}, \reset \tm) \}$. We prove that $\{
  (\shift \vark {\reset \tm}, \shift \vark \tm) \} \cup \rel \cup \Id$ is a
  big-step bisimulation. The transition $\shift \vark {\reset \tm} \lts\ctx
  \reset {\reset {\subst \tm \vark {\lam \varx {\reset{\inctx \ctx \varx}}}}}$
  is matched by $\shift \vark \tm \lts\ctx \reset {\subst \tm \vark {\lam \varx
      {\reset{\inctx \ctx \varx}}}}$, and conversely. Let $\reset {\reset \tm}
  \rel \reset \tm$. It is straightforward to check that $\reset {\reset \tm}
  \ltswk \valzero \val$ iff $\tm \ltswk \valzero \val$ iff $\reset \tm
  \ltswk\valzero \val$. Therefore, any $\ltswk\val$ transition from $\reset
  {\reset \tm}$ is matched by $\reset \tm$, and conversely. If $\reset\tm
  \lts\tau \tm'$, then $\tm'$ is a value or $\tm' = \reset{\tm''}$ for some $\tm''$;
  consequently, neither $\reset \tm$ nor $\reset{\reset \tm}$ can perform a $\lts\ctx$
  transition. \qed
\end{proof}

\begin{proposition}
  We have $\reset {\app {\lamp \varx \tmzero}{\reset \tmone}} \bisim \app
  {\lamp \varx {\reset \tmzero}}{\reset \tmone}$.
\end{proposition}

\begin{proof}
  We prove that $\{ (\reset {\app {\lamp \varx \tmzero}{\reset \tmone}}, \app
  {\lamp \varx {\reset \tmzero}}{\reset \tmone}) \} \cup \Id$ is a big-step
  bisimulation. A transition $\reset {\app {\lamp \varx \tmzero}{\reset \tmone}}
  \ltswk\act \tm'$ (with $\act \neq \tau$) is possible only if $\reset \tmone$
  evaluates to some value $\val$. In this case, we have $\reset {\app {\lamp
      \varx \tmzero}{\reset \tmone}} \ltswk\tau \reset {\app {\lamp \varx
      \tmzero} \val} \lts\tau \reset {\subst \tmzero \varx \val}$ and $\app
  {\lamp \varx {\reset \tmzero}}{\reset \tmone} \ltswk\tau \reset {\subst
    \tmzero \varx \val}$. From this, it is easy to see that $\reset {\app {\lamp
      \varx \tmzero}{\reset \tmone}} \ltswk\act \tm'$ (with $\act \neq \tau$)
  implies $\app {\lamp \varx {\reset \tmzero}}{\reset \tmone} \ltswk\act \tm'$,
  and conversely.  \qed
\end{proof}

\begin{proposition}
  \label{p:omega}
  If $\varx \notin \fv \ctx$, then $\app {\lamp \varx {\inctx \ctx \varx}} \tm
  \bisim \inctx \ctx \tm$.
\end{proposition}

\begin{proof}[Sketch]
  The complete proof, quite technical, can be found in
  Appendix~\ref{a:omega}. Let $\ctxzero$ be such that $\fv \ctxzero =
  \emptyset$. Given two families of contexts $(\ctx_1^i)_i$, $(\ctx_2^i)_i$, we
  write $\subs^\ctxzero_i$ (resp. $\subs^{\lam \varx {\inctx \ctxzero
      \varx}}_i$) the substitution mapping $\vark_i$ to $\lam y
  {\reset{\inctx{\ctx_1^i}{\inctx \ctxzero{\inctx {\ctx_2^i} y}}}}$ (resp.
  $\lam y {\reset{\inctx{\ctx_1^i}{\app {\lamp \varx {\inctx \ctxzero
            \varx}}{\inctx {\ctx_2^i} y}}}}$). We define
  \begin{eqnarray*}
    \rel_1 &=& \{ (\inctx \rctx {\app {\lamp \varx {\inctx \ctxzero
          \varx}} \tm} \subs^{\lam \varx {\inctx \ctxzero \varx}}_0 \ldots
    \subs^{\lam \varx {\inctx \ctxzero \varx}}_n, \inctx \rctx {\inctx \ctxzero
      \tm} \subs^\ctxzero_0 \ldots \subs^\ctxzero_n), \\ 
    && \hspace{19em} \fv{\tm, \rctx} \,\subseteq\, \{ \vark_0 \ldots \vark_n\} \} \\
    \rel_2 &=& \{ (\tm \subs^{\lam \varx {\inctx \ctxzero \varx}}_0 \ldots
    \subs^{\lam \varx {\inctx \ctxzero \varx}}_n, \tm \subs^\ctxzero_0 \ldots
    \subs^\ctxzero_n), \fv \tm \,\subseteq\, \{ \vark_0 \ldots \vark_n\} \}
  \end{eqnarray*}
  and we prove that $\rel_1 \cup \rel_2$ is a bisimulation. The relation
  $\rel_1$ contains the terms related by Proposition \ref{p:omega}. The
  transitions from terms in $\rel_1$ give terms in $\rel_1$, except if a capture
  happens in $\tm$; in this case, we obtain terms in $\rel_2$. Similarly, most
  transitions from terms in $\rel_2$ give terms in $\rel_2$, except if a term
  $\lam y {\reset{\inctx{\ctx_1^i}{\inctx \ctxzero{\inctx {\ctx_2^i} y}}}}$
  (resp.  $\lam y {\reset{\inctx{\ctx_1^i}{\app {\lamp \varx {\inctx \ctxzero
            \varx}}{\inctx {\ctx_2^i} y}}}}$) is applied to a value (\ie if $\tm
  = \inctx \rctx {\app {\vark_i} \val}$). In this case, the $\beta$-reduction
  generates terms in $\rel_1$. \qed
\end{proof}

\subsection{Bisimilarity and CPS Equivalence}
\label{ss:discussion-CPS}

In Section \ref{ss:example}, we have considered all the axioms of
Fig. \ref{f:axioms}, except $\AXshiftelim$. The terms $\shift \vark {\app
  \vark \tm}$ (with $\vark \notin \fv \tm$) and $\tm$ are not bisimilar in
general, as we can see in the following result.
\begin{proposition}
  We have $\shift \vark {\app \vark \val} \not\bisim \val$.
\end{proposition}
The $\lts\ctx$ transition from $\shift \vark {\app \vark \val}$ cannot
be matched by $\val$. In terms of contextual equivalence, it is not
possible to equate a stuck term and a value (it is also forbidden by
the relation $\ctxequivone$ of Section \ref{ss:context-equiv}). The
CPS equivalence cannot distinguish between stuck terms and values,
because the CPS translation turns all $\lamshift$ terms into
$\lambda$-calculus terms of the form $\lam {\vark_1\vark_2}{\tm}$,
where $\vark_1$ is the continuation up to the first enclosing reset,
and $\vark_2$ is the continuation beyond this reset. Therefore, the
CPS translation (and CPS equivalence) assumes that there is always an
enclosing reset, while contextual equivalence does not. To be in
accordance with CPS, the contextual equivalence should be changed, so
that it tests terms only in contexts with an outermost delimiter. We
conjecture that the CPS equivalence is included in such a modified
contextual equivalence. Note that stuck terms can no longer be
observed in such modified relation, because a term within a reset
cannot become stuck (see the proof of Proposition
\ref{p:reset-reset}). Therefore, the bisimilarity of this article is
too discriminative w.r.t. to this modified equivalence, and a new
complete bisimilarity has to be found.

Conversely, there exist bisimilar terms that are not CPS equivalent:
\begin{proposition}
\begin{enumerate}
\item  
\label{p:ce-cps-one}
We have $\Omega \bisim \app \Omega \Omega$, but $\Omega$ and $\app \Omega
  \Omega$ are not CPS equivalent.
\item 
\label{p:ce-cps-two}
Let $\Theta = \app \theta \theta$, where $\theta = \lam {\varx
  y}{\app y {\lamp z {\app {\app {\app \varx \varx} y} z}}}$, and
  $\Delta = \lam \varx {\app{\delta_\varx}{\delta_\varx}}$, where
  $\delta_\varx = \lam y {\app \varx {\lamp z {\app {\app y y}
      z}}}$. We have $\Theta \bisim \Delta$, but $\Theta$ and $\Delta$
  are not CPS equivalent.
\end{enumerate}
\end{proposition}
Contextual equivalence puts all diverging terms in one equivalence
class, while CPS equivalence is more discriminating. Furthermore, as
is usual with equational theories for $\lambda$-calculi, CPS
equivalence is not strong enough to equate Turing's and Curry's
(call-by-value) fixed point combinators.

\section{Conclusion}
\label{s:conclusion}

In this article, we propose a first study of the behavioral theory of
a $\lambda$-calculus with delimited-control operators. We discuss
various definitions of contextual equivalence, and we define an
LTS-based applicative bisimilarity which is sound and complete
w.r.t. the chosen contextual equivalence. Finally, we point out some
differences between bisimilarity and CPS equivalence. We believe this
work can be pursued in the following directions.

\paragraph{Up-to techniques.} Up-to techniques
\cite{Sangiorgi-Walker:01,Lassen:98,Sangiorgi-al:LICS07} have been
introduced to simplify the use of bisimulation in proofs of program
equivalences. The idea is to prove terms equivalences using relations
that are usually not bisimulations, but are included in bisimulations.
The validity of some applicative bisimilarities up-to context remains
an open problem in the $\lambda$-calculus~\cite{Lassen:98};
nevertheless, we want to see if some up-to techniques can be applied
to the bisimulations of this article.

\paragraph{Other forms of bisimilarity.} Applicative bisimilarity is simpler to
prove than contextual equivalence, but its definition still involves
some quantification over values and pure contexts in
labels. \emph{Normal bisimilarities} are easier to use because their
definitions do not feature such quantification. Lassen has developed a
notion of normal bisimilarity, sound in various $\lambda$-calculi
\cite{Lassen:LICS05,Lassen:MFPS05,Lassen-Levy:LICS08}, and also
complete in the $\lambda \mu$-calculus with state
\cite{Stoevring-Lassen:POPL07}. It would be interesting to see if this
equivalence can be defined in a calculus with delimited control, and
if it is complete in this setting. Another kind of equivalence worth
exploring is \emph{environmental bisimilarity}
\cite{Sangiorgi-al:LICS07}.

\paragraph{Other calculi with control.} Defining an applicative bisimilarity for
the call-by-name variant of $\lamshift$ and for the hierarchy of
delimited-control operators \cite{Danvy-Filinski:LFP90} should be
straightforward. We plan to investigate applicative bisimilarities for
a typed $\lamshift$ as well \cite{Biernacka-Biernacki:PPDP09}. The
problem seems more complex in calculi with abortive operators, such as
call/cc. Because there is no delimiter for capture, these languages
are not compositional (\ie $\tm \redcbv \tm'$ does not imply $\inctx
\ctx \tm \redcbv \inctx \ctx {\tm'}$), which makes the definition of a
compositional LTS more difficult. 

\paragraph*{Acknowledgments:}
We thank Ma{\l}gorzata Biernacka, Daniel Hirschkoff, Damien Pous, and the
anonymous referees for many helpful comments on the presentation of this work.

\bibliographystyle{abbrv}

\newpage
\appendix

\section{Contextual Equivalence}
\label{a:context-equiv}

We prove the result admitted in Section \ref{ss:context-equiv}.
\begin{lemma}
\label{l:aone}  
  Let $\tmzero = \app {(\shift \vark {\lam \varx \varx})} \Omega$. Let $\tm$
  be such that $\varx$ is the only free variable of $\tm$. If $\subst \tm \varx
  \tmzero \evalcbv \valzero$, then there exists $\val$ such that $\valzero =
  \subst \val \varx \tmzero$, and for all $\tmone$, we have $\subst \tm \varx
  \tmone \evalcbv \subst \val \varx \tmone$.
\end{lemma}

\begin{proof}
  We proceed by induction on the number $n$ of reduction steps in
  $\subst \tm \varx \tmzero \evalcbv \valzero$. If $n=0$, then $\tm$
  is a value, and the result is obvious. For $n>0$, we have $\subst
  \tm \varx \tmzero \redcbv \tm' \evalcbv \valzero$. By case analysis
  on the reduction $\subst \tm \varx \tmzero \redcbv \tm'$, we can see
  that $\tm'$ can be written $\subst{\tm''} \varx \tmzero$. In
  particular, the reduction $\inctx \rctx {\reset {\tmzero}} \redcbv
  \inctx \rctx {\reset {\app {\lamp \varx {\reset {\app \varx
            \Omega}}}{\lamp \varx \varx}}}$ is not possible, because
  otherwise $\tm'$ would diverge. By induction hypothesis, there
  exists $\val$ such that $\valzero = \subst \val \varx \tmzero$, and
  for all $\tmone$, we have $\subst {\tm''} \varx \tmone \evalcbv
  \subst \val \varx \tmone$. Therefore, $\subst {\tm} \varx \tmone
  \evalcbv \subst \val \varx \tmone$ holds, as wished.  \qed
\end{proof}

\section{Bisimilarity}

\subsection{Labelled Transition System}

\begin{lemma}
  \label{l:decompose-shift}
  If $\tm \lts\ctx \tm'$, then there exist $\ctx'$, $\vark$, and $\tms$
  such that $\tm = \inctx {\ctx'}{\shift \vark \tms}$ and $\tm' =
  \reset{\subst \tms \vark {\lam \varx {\reset {\inctx {\ctx}{\inctx
            {\ctx'} \varx}}}}}$.
\end{lemma}

\begin{proof}
  By induction on $\tm \lts\ctx \tm'$. Suppose we have $\shift \vark \tms
  \lts\ctx \subst \tms \vark {\lam \varx {\reset {\inctx \ctx \varx}}}$; the
  result is obvious. Suppose we have $\tm = \app \val \tmzero$ and $\tmzero
  \lts{\vctx \val \ctx} \tm'$. By induction hypothesis there exists $\ctx'$, $\tms$ such
  that $\tmzero = \inctx {\ctx'}{\shift \vark \tms}$ and $\tm' = \reset{\subst
    \tms \vark {\lam \varx {\reset {\inctx {(\vctx \val \ctx)}{\inctx {\ctx'}
            \varx}}}}}$. Therefore we have $\app \val \tmzero = \app \val
  {\inctx {\ctx'}{\shift \vark \tms}}$, and $\tm' = \reset {\subst \tms \vark
    {\lam \varx {\reset {\inctx \ctx {(\vctx \val {\inctx {\ctx'} \varx})}}}}}$,
  as required. The other cases are treated similarly.  \qed
\end{proof}

\begin{lemma}[Lemma \ref{l:lts-redcbv-main} in the article]
  \label{l:lts-redcbv}
  The following holds:
  \begin{itemize}
  \item We have $\ltstau = \redcbv$.
  \item If $\tm \lts\ctx \tm'$, then $\tm$ is a stuck term, and $\reset{\inctx
      \ctx \tm} \ltstau \tm'$.
  \item If $\tm \lts\val \tm'$, then $\tm$ is a value, and $\app \tm \val
    \ltstau \tm'$.
  \end{itemize}
\end{lemma}

\begin{proof}
  For the first result, the only difficult transition to check is the capture by
  shift. If $\reset \tm \lts\tau \tm'$ with $\tm \lts\mtctx \tm'$, then by Lemma
  \ref{l:decompose-shift}, there exists $\ctx$, $\tms$ such that $\tm = \inctx
  \ctx {\shift \vark \tms}$ and $\tm' = \reset{\subst \tms \vark {\lam \varx
      {\reset {\inctx \ctx \varx}}}}$. We have $\reset {\inctx \ctx {\shift
      \vark \tms}} \redcbv \reset{\subst \tms \vark {\lam \varx {\reset {\inctx
          \ctx \varx}}}}$, as wished. Conversely, if $\inctx \rctx{\reset {\inctx \ctx
    {\shift \vark \tms}}} \redcbv \inctx \rctx {\reset{\subst \tms \vark {\lam \varx {\reset
        {\inctx \ctx \varx}}}}}$, then one can check that $\inctx \rctx {\reset {\inctx \ctx
    {\shift \vark \tms}}} \lts\tau \inctx \rctx {\reset{\subst \tms \vark {\lam \varx {\reset
        {\inctx \ctx \varx}}}}}$ holds by induction on $\rctx$.

  For the second result, by Lemma \ref{l:decompose-shift}, there exist
  $\ctx'$, $\vark$ and $\tms$ such that $\tm = \inctx {\ctx'}{\shift
    \vark \tms}$ and $\tm' = \reset{\subst \tms \vark {\lam \varx
      {\reset {\inctx {\ctx}{\inctx {\ctx'} \varx}}}}}$. The term
  $\tm$ is stuck by Lemma \ref{l:stuck}. We can also easily check that
  $\reset{\inctx \ctx \tm} = \reset{\inctx \ctx {\inctx {\ctx'}{\shift
        \vark \tms}}} \lts\tau \reset{\subst \tms \vark {\lam \varx
      {\reset {\inctx {\ctx}{\inctx {\ctx'} \varx}}}}}$ holds.

  For the last result, by looking at the LTS, it is easy to see that $\tm$ must
  be a value $\lam \varx \tms$, and $\tm' = \subst \tms \varx \val$, and $\app
  \tm \val \lts\tau \tm'$.
  \qed
\end{proof}

\subsection{Howe's method}
\label{a:howe}

\begin{lemma}
  \label{l:same-size}
  If $\tmzero \clohbisim \tmone$, then there exists a substitution $\subs$ which
  closes $\tmzero$ and $\tmone$ such that $\tmzero \subs \clohbisimc \tmone
  \subs$, and the size of the derivation of $\tmzero \subs \clohbisimc \tmone
  \subs$ is equal to the size of the derivation of $\tmzero \clohbisim \tmone$.
\end{lemma}

\begin{proof}
  By induction on $\tmzero \clohbisim \tmone$. Suppose we have $\tmzero
  \open\bisim \tmone$. Let $\subs$ be a substitution which closes $\tmzero$ and
  $\tmone$; we have $\tmzero\subs \open\bisim \tmone\subs$. The remaining cases
  are easy using induction.  \qed
\end{proof}

\begin{lemma}
  \label{l:cloh-act}
  Let $\tm$ be a closed term. If $\tm \lts\act \tm'$, then for all $\act
  \clohbisimc \act'$, there exists $\tm''$ such that $\tm \lts{\act'} \tm''$ and
  $\tm' \clohbisimc \tm''$.
\end{lemma}

\begin{proof}
  By induction on $\tm \lts\act \tm'$.
  \qed
\end{proof}

\begin{lemma}
  \label{l:beta}
  If $\lam \varx \tmzero \clohbisimc \lam \varx \tmone$ and $\valzero
  \clohbisimc \valone$ then $\subst \tmzero \varx \valzero \clohbisimc \subst
  \tmone \varx \valone$. 
\end{lemma}

\begin{proof}
  By induction on $\lam \varx \tmzero \clohbisimc \lam \varx \tmone$.

  Suppose $\lam \varx \tmzero \bisim \lam \varx \tmone$. We have $\lam \varx
  \tmzero \lts\valzero \subst \tmzero \varx \valzero$, so by Lemma
  \ref{l:cloh-act}, there exists $\tmzero'$ such that $\lam \varx \tmzero
  \lts\valone \tmzero'$ and $\subst \tmzero \varx \valzero \clohbisim
  \tmzero'$. By bisimilarity, there exists $\tmone'$ such that $\lam \varx
  \tmone \lts\valone \tmone'$ and $\tmzero' \bisim \tmone'$. The only possible
  outcome is $\tmone' = \subst \tmone \varx \valone$, therefore we have $\subst
  \tmzero \varx \valzero \clohbisim \bisim \subst \tmone \varx \valone$. Because
  the considered terms are closed, we have the required result.

  Suppose $\lam \varx \tmzero \clohbisim \bisim \lam \varx \tmone$. The result
  can easily be proved using induction and a similar reasoning as in the first
  case.

  Suppose $\lam \varx \tmzero \comp\clohbisim \lam \varx \tmone$. Then we have
  $\tmzero \clohbisim \tmone$; therefore, by Lemma \ref{l:properties}, we have
  $\subst \tmzero \varx \valzero  \clohbisim \subst \tmone \varx \valone$, as
  required. 
  \qed
\end{proof}

\begin{lemma}
  \label{l:cloh-val}
  If $\valzero \clohbisimc \tmone$, then there exists $\valone$ such that
  $\tmone \ltswk\tau \valone$ and $\valzero \clohbisimc \valone$.
\end{lemma}

\begin{proof}
  By induction on $\valzero \clohbisimc \tmone$.

  Suppose $\valzero \bisim \tmone$. Let $\valzero = \lam \varx
  \tmzero$; for all $\val$, we have $\valzero \lts\val \subst \tmzero
  \varx \val$. By bisimilarity, there exists $\tmone^\val$ such that
  $\valzero \ltswk\val \tmone^\val$ and $\subst \tmzero \varx \val
  \bisim \tmone^\val$. Therefore there exists $\valone=\lam \varx
         {\tmone'}$ such that $\tmone \ltswk\tau \valone \lts\val
         \subst {\tmone'} \varx \val$ and $\subst \tmzero \varx \val
         \bisim \subst {\tmone'} \varx \val$. Because this holds for
         all $\val$, we have $\tmzero \open\bisim \tmone'$, therefore
         we have $\tmzero \clohbisim \tmone'$. From this observation,
         we deduce $\lam \varx \tmzero \clohbisimc \lam \varx
         {\tmone'}$, as wished.

  The case $\valzero \clohbisim \bisim \tmone$ relies on induction and a similar
  reasoning as in the first case. Suppose $\valzero \comp\clohbisim \tmone$. Let
  $\valzero = \lam \varx \tmzero$; we have $\tmone = \lam \varx {\tmone'}$ with
  $\tmzero \clohbisim \tmone'$. Because $\clohbisim$ is a congruence, we have
  $\lam \varx \tmzero \clohbisim \lam \varx {\tmone'}$, as wished.
  \qed
\end{proof}

\begin{lemma}
  \label{l:inctx}
  If $\ctx_0 \clohbisim \ctx_1$ and $\tmzero \clohbisim \tmone$ then $\inctx
  {\ctx_0} \tmzero \clohbisim \inctx {\ctx_1} \tmone$.
\end{lemma}

\begin{proof}
  By induction on $\ctx_0 \clohbisim \ctx_1$. \qed
\end{proof}

\begin{lemma}[Lemma \ref{l:sim-property-main} in the article]
  \label{l:sim-property}
  If $\tmzero \clohbisimc \tmone$ and $\tmzero \lts\act \tmzero'$, then for all
  $\act \clohbisimc \act'$, there exists $\tmone'$ such that $\tmone
  \ltswk{\act'} \tmone'$ and $\tmzero' \clohbisimc \tmone'$.
\end{lemma}

\begin{proof}
  By induction on the size of the derivation of $\tmzero \clohbisimc \tmone$.

  If $\tmzero \open\bisim \tmone$, then we have $\tmzero \bisim \tmone$ because
  we work with closed terms. By Lemma \ref{l:cloh-act}, there exists $\tmzero''$
  such that $\tmzero \lts{\act'} \tmzero''$ and $\tmzero' \clohbisimc
  \tmzero''$. By bisimilarity, there exists $\tmone'$ such that $\tmone
  \ltswk{\act'} \tmone'$ and $\tmzero'' \bisim \tmone'$ (\ie $\tmzero'' \open\bisim
  \tmone'$ because the terms are closed). Therefore we have $\tmzero'
  \clohbisimc \open\bisim \tmone'$, \ie $\tmzero' \clohbisimc \tmone'$, as
  required.\\

  If $\tmzero \clohbisim \tm_2 \open\bisim \tmone$, then by Lemma
  \ref{l:same-size}, there exists $\subs$ such that $\tmzero \subs \clohbisimc
  \tm_2 \subs$ and the size of the derivation of $\tmzero \subs \clohbisimc
  \tm_2 \subs$ is the same as for $\tmzero \clohbisimc \tm_2$. Because $\tmzero$
  and $\tmone$ are closed, and by definition of $\open \bisim$, we have in fact
  $\tmzero \clohbisimc \tm_2 \subs \bisim \tmone$. By induction hypothesis,
  there exists $\tm'_2$ such that $\tm_2 \subs \ltswk{\act'} \tm'_2$ and
  $\tmzero' \clohbisimc \tm_2'$. By bisimilarity, there exists $\tmone'$ such
  that $\tmone \ltswk{\act'} \tmone'$ and $\tm_2' \bisim \tmone'$ (\ie $\tm_2'
  \open\bisim \tmone'$ because the terms are closed). Therefore we have
  $\tmzero' \clohbisimc \open\bisim \tmone'$, \ie $\tmzero'
  \clohbisimc \tmone'$, as required.\\

  If $\tmzero \comp\clohbisim \tmone$, then we distinguish several cases,
  depending on the outermost operator.
  
  Suppose $\tmzero = \lam \varx {\tms_0}$ and $\tmone = \lam \varx {\tms_1}$
  with $\tms_0 \clohbisim \tms_1$. The only possible transition is $\tmzero
  \lts\val \subst {\tms_0} \varx \val$. We have $\tmone \lts{\val'} \subst
  {\tms_1} \varx {\val'}$. By Lemma \ref{l:properties}, we have $\subst {\tms_0}
  \varx \val \clohbisim \subst {\tms_1} \varx {\val'}$, and because $\varx$ is
  the only free variable of $\tms_0$ and $\tms_1$, we have $\subst {\tms_0}
  \varx \val \clohbisimc \subst {\tms_1} \varx {\val'}$, as required.

  Suppose $\tmzero = \shift \vark {\tms_0}$ and $\tmone = \shift \vark {\tms_1}$
  with $\tms_0 \clohbisim \tms_1$. The only possible transition is $\tmzero
  \lts\ctx \reset {\subst {\tms_0} \vark {\lam \varx {\reset {\inctx \ctx
          \varx}}}}$. We have $\tmone \lts{\ctx'} \reset {\subst {\tms_1} \vark
    {\lam \varx {\reset {\inctx {\ctx'} \varx}}}}$. Because $\clohbisim$ is a
  congruence and by Lemma \ref{l:inctx}, we have $\lam \varx {\reset{\inctx \ctx
      \varx}} \clohbisimc \lam \varx {\reset{\inctx {\ctx'} \varx}}$. Therefore,
  by Lemma \ref{l:properties}, we have $\reset {\subst {\tms_0} \vark {\lam
      \varx {\reset {\inctx \ctx \varx}}}} \clohbisim \subst {\tms_1} \vark
  {\lam \varx {\reset {\inctx {\ctx'} \varx}}}$, and because $\vark$ is the only
  free variable of $\tms_0$ and $\tms_1$, we have $\reset{\subst {\tms_0} \vark
    {\lam \varx {\reset {\inctx \ctx \varx}}}} \clohbisimc \reset{\subst
    {\tms_1} \vark {\lam \varx {\reset {\inctx {\ctx'} \varx}}}}$, as required.

  Suppose $\tmzero = \app {\lamp \varx {\tms_0}} \valzero$ and $\tmone = \app
  {\tmone^1}{\tmone^2}$ with $\lam \varx {\tms_0} \clohbisimc \tmone^1$ and
  $\valzero \clohbisimc \tmone^2$. The only possible transition is $\tmzero
  \ltstau \subst {\tms_0} \varx \valzero$. By Lemma \ref{l:cloh-val}, there
  exists $\lam \varx {\tms_1}$, $\valone$ such that $\tmone^1 \ltswk\tau\lam
  \varx {\tms_1}$, $\tmone^2 \ltswk\tau \valone$, $\lam \varx {\tms_0}
  \clohbisimc \lam \varx {\tms_1}$, and $\valzero \clohbisimc \valone$. From
  $\tmone^1 \ltswk\tau\lam \varx {\tms_1}$ and $\tmone^2 \ltswk\tau \valone$, we
  can deduce $\tmone \ltswk\tau \subst {\tms_1} \varx \valone$, and from $\lam
  \varx {\tms_0} \clohbisimc \lam \varx {\tms_1}$ and $\valzero \clohbisimc
  \valone$, we have $\subst {\tms_0} \varx \valzero \clohbisim \subst {\tms_1}
  \varx \valone$ by Lemma \ref{l:beta}. Hence, we have the required result.

  Suppose $\tmzero = \app \valzero {\tms_0}$ and $\tmone = \app
  {\tmone^1}{\tms_1}$ with $\valzero \clohbisimc \tmone^1$ and $\tms_0
  \clohbisimc \tms_1$. The only possible transition is $\tmzero
  \ltstau \app \valzero {\tms'_0}$, where $\tms_0 \lts\tau
  \tms'_0$. By Lemma \ref{l:cloh-val}, there exists $\valone$ such
  that $\tmone^1 \ltswk\tau \valone$ and $\valzero \clohbisimc
  \valone$. By induction hypothesis, there exists $\tms_1'$ such that
  $\tms_1 \ltswk\tau \tms'_1$ and $\tms_1 \clohbisimc
  \tms'_1$. Therefore we have $\tmone \ltswk\tau \app \valone
       {\tms'_1}$, and because $\clohbisimc$ is a congruence, we have
       $\app \valzero {\tms'_0} \clohbisimc \app \valone {\tms'_1}$,
       hence the result holds.

  Suppose $\tmzero = \app {\tmzero^1}{\tmzero^2}$ and $\tmone = \app
  {\tmone^1}{\tmone^2}$ with $\tmzero^1 \clohbisimc \tmone^1$ and $\tmzero^2
  \clohbisimc \tmone^2$. The only possible transition is $\tmzero \ltstau \app
  {{\tmzero^1}'}{\tmzero^2}$, where $\tmzero^1 \lts\tau {\tmzero^1}'$. By
  induction, there exists ${\tmone^1}'$ such that $\tmone^1 \ltswk\tau
  {\tmone^1}'$ and ${\tmzero^1}' \clohbisimc {\tmone^1}'$. Therefore we have
  $\tmone \ltswk\tau \app {{\tmone^1}'}{\tmone^2}$, and because $\clohbisimc$ is
  a congruence, we have $\app {{\tmzero^1}'}{\tmzero^2} \clohbisimc \app
  {{\tmone^1}'}{\tmone^2}$, hence the result holds. The case $\tmzero = \reset
  {\tms_0}$, where $\tms_0$ is not a value, is treated similarly.

  Suppose $\tmzero = \reset \valzero$ and $\tmone = \reset {\tms_1}$ with
  $\valzero \clohbisimc {\tms_1}$. The only possible transition is $\tmzero
  \lts\tau \valzero$. By Lemma \ref{l:cloh-val}, there exists $\valone$ such
  that $\tms_1 \ltswk\tau \valone$ and $\valzero \clohbisimc \valone$. We have
  $\tmone \ltswk\tau \valone$, with $\valzero \clohbisimc \valone$, hence the
  result holds.
  \qed
\end{proof}

\begin{lemma}
  The relation $\rtclo{(\clohbisimc)}$ is a bisimulation.
\end{lemma}

\begin{proof}
  We know that $\rtclo{(\clohbisimc)}$ is symmetric by Lemma \ref{l:properties},
  so it is enough to prove that $\rtclo{(\clohbisimc)}$ is a simulation. Let
  $\tmzero \rtclo{(\clohbisimc)} \tmone$; there exists an integer $k$ such that
  $\tmzero \mathrel{(\clohbisimc)^k} \tmone$. Let $(\tmzero^i)_{i \in \{1 \ldots k\}}$ be
  the terms such that $\tmzero \clohbisimc \tmzero^1 \clohbisimc \tmzero^2
  \ldots \tmzero^{k-1} \clohbisimc \tmzero^k=\tmone$. Let $\tmzero \lts\act
  \tmzero'$. We prove by induction on $i \in 1 \ldots k$ that there exists
  ${\tmzero^i}'$ such that $\tmzero^i \ltswk\act {\tmzero^i}'$ and $\tmzero'
  \mathrel{(\clohbisimc)^i} {\tmzero^i}'$. Suppose $i=1$. We have $\act \clohbisimc \act$
  so by Lemma \ref{l:sim-property}, there exists ${\tmzero^1}'$ such that
  $\tmzero^1 \ltswk\act {\tmzero^1}'$ and $\tmzero' \clohbisimc {\tmzero^1}'$,
  as wished. The case $1 < i \leq k$ is easy by induction. In particular, for
  $i=k$, we have $\tmone = \tmzero^k \ltswk\act {\tmzero^k}'$ and $\tmzero'
  (\clohbisimc)^k {\tmzero^k}'$. We have the required result because
  $(\clohbisimc)^k \,\subseteq\, \rtclo{(\clohbisimc)}$.
  \qed
\end{proof}

\begin{theorem}[Theorem \ref{t:soundness-main} in the article]
  \label{t:soundness}
  We have $\bisim \,\subseteq\, \ctxequiv$.
\end{theorem}

\begin{proof}
  Let $\tmzero \bisim \tmone$, and $\cctx$ a context. Because $\bisim$ is a
  congruence, we have $\inctx \cctx \tmzero \bisim \inctx \cctx \tmone$. If
  $\inctx \cctx \tmzero \evalcbv \valzero$, then we have $\inctx \cctx \tmzero
  \ltswk\tau \valzero \lts\val$; by bisimilarity, there exists $\valone$ such
  that $\inctx \cctx \tmone \ltswk\tau \valone \lts\val$, therefore we have
  $\inctx \cctx \tmone \evalcbv \valone$. The reasoning is the same if $\inctx
  \cctx \tmzero \evalcbv \tmzero'$, where $\tmzero'$ is stuck, and for the
  evaluations of $\inctx \cctx \tmone$.
  \qed
\end{proof}

\subsection{Completeness}
\label{a:completeness}

\begin{theorem}[Theorem \ref{t:completeness-main} in the article]
  \label{t:completeness}
  We have $\rctxequiv \,\subseteq\, \bisim$.
\end{theorem}

\begin{proof}
  Because $\rctxequiv$ is symmetric, it is enough to prove that $\rctxequiv$ is a
  big-step simulation. Let $\tmzero \rctxequiv \tmone$. We have two cases to
  consider.

  Suppose $\tmzero \ltswk\val \tmzero'$. Then we have $\tmzero
  \evalcbv \valzero$ for some $\valzero$. By definition of $\rctxequiv$,
  there exists $\valone$ such that $\tmone \evalcbv
  \valone$. Therefore, we have $\tmone \ltswk\val \tmone'$ for some
  $\tmone'$ and $\app \tmzero \val \ltswk\tau \tmzero'$ and $\app
  \tmone \val \ltswk\tau \tmone'$ by Lemma \ref{l:lts-redcbv}. Hence,
  we have $\app \tmzero \val \bisim \tmzero'$ and $\app \tmone \val
  \bisim \tmone'$ by Lemma \ref{l:bisim-eval}. Because $\tmzero
  \rctxequiv \tmone$, we have $\app \tmzero \val \rctxequiv \app
  \tmone \val$. Finally, we have $\tmzero' \bisim \app \tmzero \val
  \rctxequiv \app \tmone \val \bisim \tmone'$ and $\tmzero' \rctxequiv
  \tmone'$ by Theorem \ref{t:soundness} and transitivity of
  $\rctxequiv$.

  Suppose $\tmzero \ltswk\ctx \tmzero'$. Then we have $\tmzero \evalcbv
  \tmzero''$ for some $\tmzero''$. By definition of $\rctxequiv$, there exists
  $\tmone''$ such that $\tmone \evalcbv \tmone''$. Therefore, by definition of
  the LTS, we have $\tmone \ltswk\ctx \tmone'$ for some $\tmone'$. We then have
  $\reset{\inctx \ctx \tmzero} \ltswk\tau \tmzero'$ and $\reset{\inctx \ctx
    \tmone} \ltswk\tau \tmone'$ by Lemma \ref{l:lts-redcbv}. Hence, we have
  $\reset{\inctx \ctx \tmzero} \bisim \tmzero'$ and $\reset{\inctx \ctx \tmone}
  \bisim \tmone'$ by Lemma \ref{l:bisim-eval}. Because $\tmzero \rctxequiv
  \tmone$, we have $\reset{\inctx \ctx \tmzero} \rctxequiv \reset{\inctx \ctx
    \tmone}$. Finally, we have $\tmzero' \bisim \reset{\inctx \ctx \tmzero}
  \rctxequiv \reset{\inctx \ctx \tmone} \bisim \tmone'$, and, therefore we have
  $\tmzero' \rctxequiv \tmone'$ by Theorem \ref{t:soundness} and transitivity of
  $\rctxequiv$.  \qed
\end{proof}

\section{Proof of Proposition \ref{p:omega}}
\label{a:omega}

\begin{proof}
  Let $\ctxzero$ be such that $\fv \ctxzero = \emptyset$. We let $\subs^\ctxzero_i$
  (resp. $\subs^{\lam \varx {\inctx \ctxzero \varx}}_i$) range over
  substitutions mapping $\vark_i$ to
  $\lam y {\reset{\inctx{\ctx_1^i}{\inctx \ctxzero{\inctx {\ctx_2^i} y}}}}$ (resp.
  $\lam y {\reset{\inctx{\ctx_1^i}{\app {\lamp \varx {\inctx \ctxzero
            \varx}}{\inctx {\ctx_2^i} y}}}}$) for some fixed $\ctx_1^i$, $\ctx_2^i$. We define
  \begin{eqnarray*}
    \rel_1 &=& \{ (\inctx \rctx {\app {\lamp \varx {\inctx \ctxzero
          \varx}} \tm} \subs^{\lam \varx {\inctx \ctxzero \varx}}_0 \ldots
    \subs^{\lam \varx {\inctx \ctxzero \varx}}_n, \inctx \rctx {\inctx \ctxzero
      \tm} \subs^\ctxzero_0 \ldots \subs^\ctxzero_n), \\ 
    && \hspace{19em} \fv{\tm, \rctx} \,\subseteq\, \{ \vark_0 \ldots \vark_n\} \} \\
    \rel_2 &=& \{ (\tm \subs^{\lam \varx {\inctx \ctxzero \varx}}_0 \ldots
    \subs^{\lam \varx {\inctx \ctxzero \varx}}_n, \tm \subs^\ctxzero_0 \ldots
    \subs^\ctxzero_n), \fv \tm \,\subseteq\, \{ \vark_0 \ldots \vark_n\} \}
  \end{eqnarray*}
  and we prove that $\rel_1 \cup \rel_2$ is a bisimulation. First, let $\tmzero \rel_1
  \tmone$ with $\tmzero = \inctx \rctx {\app {\lamp \varx {\inctx \ctxzero
        \varx}} \tm} \subs^{\lam \varx {\inctx \ctxzero \varx}}_0 \ldots
  \subs^{\lam \varx {\inctx \ctxzero \varx}}_n$ and $\tmone = \inctx \rctx
  {\inctx \ctxzero \tm} \subs^\ctxzero_0 \ldots \subs^\ctxzero_n$; we consider
  the possible transitions from $\tmzero$.
  \begin{itemize}
  \item If $\tm \lts\tau \tm'$, then $\tmzero \lts\tau \tmzero' = \inctx \rctx
    {\app {\lamp \varx {\inctx \ctxzero \varx}}{\tm'}} \subs^{\lam \varx {\inctx
        \ctxzero \varx}}_0 \ldots \subs^{\lam \varx {\inctx \ctxzero \varx}}_n$
    and $\tmone \lts\tau \tmone' = \inctx \rctx {\inctx \ctxzero {\tm'}}
    \subs^\ctxzero_0 \ldots \subs^\ctxzero_n$. We still have $\tmzero' \rel_1
    \tmone'$.
  \item If $\tm$ is a value $\val$, then $\tmzero \lts\tau \tmzero' = \inctx
    \rctx {\inctx \ctxzero \val} \subs^{\lam \varx {\inctx \ctxzero \varx}}_0
    \ldots \subs^{\lam \varx {\inctx \ctxzero \varx}}_n$ (with rule $\LTSbeta$),
    and $\tmone = \inctx \rctx {\inctx \ctxzero \val} \subs^\ctxzero_0 \ldots
    \subs^\ctxzero_n$. We have $\tmzero' \rel_2 \tmone$.
  \item If $\tm$ is a stuck term, then we have two possibilities. If $\rctx$
    contains a reset, then the shift in $\tm$ can capture the context up to this
    delimiter, resulting in $\tmzero \lts\tau \tmzero' = \inctx {\rctx'}
    {\subst{\tm'}{\vark_{n+1}}{\lam y {\reset{\inctx{\ctx_1}{\app {\lamp \varx
                {\inctx \ctxzero \varx}}{\inctx {\ctx_2} y}}}}}} \subs^{\lam
      \varx {\inctx \ctxzero \varx}}_0 \ldots \subs^{\lam \varx {\inctx \ctxzero
        \varx}}_n$ and $\tmone \lts\tau \tmone' = \inctx {\rctx'}
    {\subst{\tm'}{\vark_{n+1}}{\lam y {\reset{\inctx{\ctx_1}{\inctx \ctxzero
              {\inctx {\ctx_2} y}}}}}} \subs^{\ctxzero}_0 \ldots
    \subs^{\ctxzero}_n$. Otherwise, $\tm$ can capture $\rctx$, giving $\tmzero
    \lts\ctx \tmzero' = \subst{\tm'}{\vark_{n+1}}{\lam y
      {\reset{\inctx{\ctx_1}{\app {\lamp \varx {\inctx \ctxzero \varx}}{\inctx
              {\ctx_2} y}}}}} \subs^{\lam \varx {\inctx \ctxzero \varx}}_0
    \ldots \subs^{\lam \varx {\inctx \ctxzero \varx}}_n$, and similarly $\tmone
    \lts\ctx \tmone' = \subst{\tm'}{\vark_{n+1}}{\lam y
      {\reset{\inctx{\ctx_1}{\inctx \ctxzero {\inctx {\ctx_2} y}}}}}
    \subs^{\ctxzero}_0 \ldots \subs^{\ctxzero}_n$. In both cases, we have
    $\tmzero' \rel_2 \tmone'$.
  \end{itemize}
  Conversely, one can check that the transitions from $\tmone$ are matched by
  $\tmzero$. Now, we consider $\tmzero \rel_2 \tmone$ with $\tmzero = \tm
  \subs^{\lam \varx {\inctx \ctxzero \varx}}_0 \ldots \subs^{\lam \varx {\inctx
      \ctxzero \varx}}_n$ and $\tmone = \tm \subs^\ctxzero_0 \ldots
  \subs^\ctxzero_n$. We enumerate the possible transitions from $\tmzero$.
  \begin{itemize}
  \item If $\tm \lts\act \tm'$, then $\tmzero \lts\act \tmzero' = \tm'
    \subs^{\lam \varx {\inctx \ctxzero \varx}}_0 \ldots \subs^{\lam \varx
      {\inctx \ctxzero \varx}}_n$ and $\tmone \lts\act \tmone'= \tm'
    \subs^\ctxzero_0 \ldots \subs^\ctxzero_n)$. We still have $\tmzero' \rel_2
    \tmone'$.
  \item Suppose $\tm = \inctx \rctx {\app {\lamp z {\tm'}}{\vark_i}}$. Then,
    using rule $\LTSbeta$, we have the transitions $\tmzero \lts\tau \tmzero' =
    \inctx \rctx {\subst {\tm'} z {\lam y {\reset{\inctx {\ctx_1^i}{\app {\lamp
                \varx {\inctx \ctxzero \varx}}{\inctx {\ctx_2^i} y}}}}}}
    \subs^{\lam \varx {\inctx \ctxzero \varx}}_0 \ldots \subs^{\lam \varx
      {\inctx \ctxzero \varx}}_n$ and $\tmone \lts\tau \tmone' = \inctx \rctx
    {\subst {\tm'} z {\lam y {\reset{\inctx {\ctx_1^i}{\inctx \ctxzero {\inctx
                {\ctx_2^i} y}}}}}} \subs^\ctxzero_0 \ldots
    \subs^\ctxzero_n$. The terms $\tmzero'$ and $\tmone'$ can be written $\inctx
    \rctx {\subst {\tm'} z {\vark_i}} \subs^{\lam \varx {\inctx \ctxzero
        \varx}}_0 \ldots \subs^{\lam \varx {\inctx \ctxzero \varx}}_n$ and
    $\inctx \rctx {\subst {\tm'} z {\vark_i}} \subs^\ctxzero_0 \ldots
    \subs^\ctxzero_n$, therefore we have $\tmzero' \rel_2 \tmone'$. The
    reasoning is the same if $\tm = \inctx \rctx {\reset{\vark_i}}$.
  \item If $\tm = \inctx \rctx {\app {\vark_i} \val}$, then $\tmzero \lts\tau
    \tmzero' = \inctx \rctx {\reset{\inctx {\ctx_1^i}{\app {\lamp \varx
              {\inctx \ctxzero \varx}}{\inctx {\ctx_2^i} \val}}}} \subs^{\lam
      \varx {\inctx \ctxzero \varx}}_0 \ldots \subs^{\lam \varx {\inctx \ctxzero
        \varx}}_n$, and $\tmone \lts\tau \tmone' = \inctx \rctx {\reset{\inctx
        {\ctx_1^i}{\inctx \ctxzero {\inctx {\ctx_2^i} \val}}}} 
    \subs^\ctxzero_0 \ldots \subs^\ctxzero_n$. We have $\tmzero' \rel_1
    \tmone'$.
  \end{itemize}
  Conversely, one can check that the transitions from $\tmone$ are
  matched by $\tmzero$. Finally, $\rel_1 \cup \rel_2$ is a
  bisimulation, as wished.  \qed
\end{proof}

\end{document}